\documentclass[11pt,twoside]{article}

\usepackage{parskip}

\usepackage{fullpage}
\usepackage{comment}

\usepackage{epsf}
\usepackage{fancyheadings}
\usepackage{graphics}
\usepackage{graphicx}
\usepackage{psfrag}
\usepackage{color}

\usepackage{bbold}

\usepackage{amsthm}
\usepackage{amsfonts}
\usepackage{amsmath}
\usepackage{amssymb}

\usepackage{graphicx,epsfig,graphics,epsf}
\usepackage{float,amsmath,amsfonts,amssymb,amsthm}
\usepackage{mathtools}
\usepackage{subcaption}

\usepackage{amsmath,amssymb}
\usepackage{amsthm}
\usepackage{mathtools}
\usepackage[noend]{algorithmic}
\usepackage[ruled,vlined]{algorithm2e}
\usepackage{url}
\usepackage{fullpage}
\usepackage{makeidx}
\usepackage{enumerate}
\usepackage[top=0.9in, bottom=0.9in, left=0.9in, right=0.9in]{geometry}
\usepackage{graphicx,float,psfrag,epsfig,caption}
\usepackage[usenames,dvipsnames,svgnames,table]{xcolor}
\definecolor{darkgreen}{rgb}{0.0,0,0.9}

\theoremstyle{plain}

\newtheorem{innercustomgeneric}{\customgenericname}
\providecommand{\customgenericname}{}
\newcommand{\newcustomtheorem}[2]{%
  \newenvironment{#1}[1]
  {%
   \renewcommand\customgenericname{#2}%
   \renewcommand\theinnercustomgeneric{##1}%
   \innercustomgeneric
  }
  {\endinnercustomgeneric}
}

\newcustomtheorem{customthm}{Theorem}
\newcustomtheorem{customlemma}{Lemma}
\newcustomtheorem{customproposition}{Proposition}
\newcustomtheorem{customcorollary}{Corollary}

\newlength{\widebarargwidth}
\newlength{\widebarargheight}
\newlength{\widebarargdepth}

\makeatletter
\long\def\@makecaption#1#2{
        \vskip 0.8ex
        \setbox\@tempboxa\hbox{\small {\bf #1:} #2}
        \parindent 1.5em  
        \dimen0=\hsize
        \advance\dimen0 by -3em
        \ifdim \wd\@tempboxa >\dimen0
                \hbox to \hsize{
                        \parindent 0em
                        \hfil 
                        \parbox{\dimen0}{\def\baselinestretch{0.96}\small
                                {\bf #1.} #2
                                } 
                        \hfil}
        \else \hbox to \hsize{\hfil \box\@tempboxa \hfil}
        \fi
        }
\makeatother

\long\def\comment#1{}

\newcommand{\Prob}{\ensuremath{{\mathbb{P}}}}

\DeclareSymbolFont{rsfs}{U}{rsfs}{m}{n}
\DeclareSymbolFontAlphabet{\mathscrsfs}{rsfs}

\numberwithin{equation}{section}

\newtheoremstyle{myexample} 
    {\topsep}                    
    {\topsep}                    
    {\rm }                   
    {}                           
    {\bf }                   
    {.}                          
    {.5em}                       
    {}  

\newtheoremstyle{myremark} 
    {\topsep}                    
    {\topsep}                    
    {\rm}                        
    {}                           
    {\bf}                        
    {.}                          
    {.5em}                       
    {}  

\newtheorem{claim}{Claim}[section]
\newtheorem{lemma}[claim]{Lemma}

\newtheorem{theorem}{Theorem}

\newtheorem{definition}[claim]{Definition}

\theoremstyle{myremark}

\usepackage{tikz}
\usetikzlibrary{shapes, arrows, calc, positioning,matrix}
\tikzset{
data/.style={circle, draw, text centered, minimum height=3em ,minimum width = .5em, inner sep = 2pt},
empty/.style={circle, text centered, minimum height=3em ,minimum width = .5em, inner sep = 2pt},
}
\usepackage{pgfplots}
\pgfplotsset{compat=1.5}
\usepgfplotslibrary{fillbetween}
\usetikzlibrary{patterns}

\DeclarePairedDelimiterX{\inp}[2]{\langle}{\rangle}{#1, #2}

\newcommand{\G}{{\sf G}}
\newcommand{\PCC}{\emph{Planted Clique Conjecture}}
\newcommand{\A}{\mathcal{A}}
\newcommand{\PCD}{{\sf{PC_D}}(n,k)}
\newcommand{\PCR}{{\sf{PC_R}}(n,k)}
\newcommand{\Binnhalf}{{\sf Bin}\left(n,\frac{1}{2}\right)}

\usepackage{bbm}
\usepackage[inline]{enumitem}
\usepackage{enumitem}

\usepackage[pagebackref,letterpaper=true,colorlinks=true,pdfpagemode=none,citecolor=OliveGreen,linkcolor=BrickRed,urlcolor=BrickRed,pdfstartview=FitH]{hyperref}

\title{Is the space complexity of planted clique recovery the same as that of detection?}
\author{Jay Mardia\thanks{Department of Electrical Engineering, Stanford
    University. jmardia@stanford.edu}}
\begin{document}

\date{}
%
\maketitle
\begin{abstract}
We study the planted clique problem in which a clique of size $k$ is planted in
an Erd\H{o}s-R\'enyi graph $G(n, \frac{1}{2})$, and one is interested
in either detecting or recovering this planted clique. This problem is interesting because it is widely believed to show a statistical-computational gap at clique size $k=\Theta(\sqrt{n})$, and has emerged as the prototypical problem with such a gap from which average-case hardness of other statistical problems can be deduced. It also displays a tight computational connection between the detection and recovery variants, unlike other problems of a similar nature. This wide investigation into the computational complexity of the planted clique problem has, however, mostly focused on its time complexity. To begin investigating the robustness of these statistical-computational phenomena to changes in our notion of computational efficiency, we ask-
\begin{quote}
\textit{Do the statistical-computational phenomena that make the planted clique an interesting problem also hold when we use `space efficiency' as our notion of computational efficiency?}
\end{quote}
It is relatively easy to show that a positive answer to this question depends on the existence of a $O(\log n)$ space algorithm that can recover planted cliques of size $k = \Omega(\sqrt{n})$. Our main result comes very close to designing such an algorithm. We show that for $k=\Omega(\sqrt{n})$, the recovery problem can be solved in $O\left(\left(\log^*{n}-\log^*{\frac{k}{\sqrt{n}}}\right) \cdot \log n\right)$ bits of space.
\begin{enumerate}
	\item If $k = \omega(\sqrt{n}\log^{(\ell)}n)$ \footnote{Here $\log^{(\ell)}n$ means we repeatedly take the logarithm of $n$ $\ell$ times. For example, $\log^{(3)}n = \log\log\log n$.} for any constant integer $\ell > 0$, the space usage is $O(\log n)$ bits.
	\item If $k = \Theta(\sqrt{n})$, the space usage is $O(\log^*{n} \cdot \log n)$ bits.
\end{enumerate}

Our result suggests that there does exist an $O(\log n)$ space algorithm to recover cliques of size $k =\Omega(\sqrt{n})$, since we come very close to achieving such parameters. This provides evidence that the statistical-computational phenomena that (conjecturally) hold for planted clique time complexity also (conjecturally) hold for space complexity.

\end{abstract}

\setcounter{tocdepth}{2}

\section{Introduction}
\label{sec:Introduction}
The planted clique problem is a well-studied task in average-case computational complexity, in which a clique of size $k$ is planted in an Erd\H{o}s-R\'enyi graph of size $n$, $G(n, \frac{1}{2})$. The problem comes in two flavours, detection ($\PCD$) and recovery ($\PCR$). In the former, we are given either a $G(n, \frac{1}{2})$ graph or a planted clique graph and must identify the graph we have been given. That is, we must detect whether or not the graph has a planted clique. In the latter, we are given a planted clique graph and must recover all the vertices in the clique.

The planted clique problem shows a variety of interesting phenomena in its time complexity. Not only does it exhibit a statistical-computational gap at clique size $k = \Theta(\sqrt{n})$, it has also emerged as the central problem whose average-case hardness implies average-case hardness for many other problems with statistical-computational gaps. See \cite{brennan2018reducibility,brennan2020reducibility} for some examples. Further, the detection and recovery problems have the same threshold at which a polynomial time statistical-computational gap shows up, even though a priori the latter could be a harder problem than the former. In fact, for several other problems such as community detection/recovery in the stochastic block model \cite{abbe2017community} or planted submatrix detection/recovery \cite{hajek2015computational,chen2014statistical}, there does indeed appear to be a difference between the time complexity of detection and recovery. They become polynomial time feasible at different signal-to-noise ratios, and this makes the lack of a gap between detection and recovery in planted clique all the more noteworthy.

Algorithmic progress on planted cliques has shown that both the detection and recovery problems can be solved `computationally efficiently' (i.e. in polynomial time) for large cliques of size $k = \Omega(\sqrt{n})$ and less efficiently in quasi-polynomial time $n^{O(\log n)}$ for cliques larger than the information-theoretic threshold, $k \geq (2+\epsilon) \log n$. The widely believed $\PCC$ even states that if the clique size is small $k = O(n^{\frac{1}{2}-\delta})$ for any constant $\delta >0$, no polynomial time algorithm can solve the planted clique detection (and hence also the recovery) problem. We survey the results providing evidence for this conjecture in Section~\ref{sec:related_work}. 

However, we do not know how robust these statistical-computational phenomena are to changes in our notion of `computational efficiency'. To begin investigating this, we ask the following question:

\begin{quote}
\textit{Do the statistical-computational phenomena that make the planted clique an interesting problem also hold when we use `space efficiency' as our notion of computational efficiency?}
\end{quote}
To answer this question, we must first discuss what a `space efficient' algorithm is. One of the most well studied classes of space bounded computation is that of logarithmic space, and it is widely considered a benchmark of `space efficient' computation.

Let us further motivate this target space complexity. For deterministic algorithms, this is the class that runs using $O(\log n)$ bits of space on inputs of size poly($n$). It is well known that any deterministic algorithm that uses at most $s(n)$ bits of space must also run in time $2^{O(s(n)+\log n)}$ \cite[Theorem 4.3]{arora2009computational}\footnote{Strictly speaking, the theorem we point to relating deterministic space complexity to time complexity \cite[Theorem 4.3]{arora2009computational} is for Turing machines. While it is convenient to define computational complexity classes using Turing machines, it is extremely inconvenient to design algorithms using them. Instead, we work with a slightly stronger model of computation that allows random access to the input to make algorithm design reasonable. However, the idea behind \cite[Theorem 4.3]{arora2009computational} also holds in any reasonable RAM model and so we ignore this distinction for the purposes of our discussion.}. This means that deterministic logspace algorithms are a subset of polynomial time algorithms, lending support to the belief that they are a good proxy for `efficient' computation. For algorithms that can use randomness, defining an appropriate notion of space bounded computation involves restricting algorithms that use $s(n)$ bits of space to at most $2^{O(s(n)+\log n)}$ time\footnote{See, for example, the section on Randomized Space-Bounded Computation in \cite{arora2009computational} or the discussion about $BP_HSPACE$ in \cite{saks1996randomization}. Note that all the algorithms we discuss in this work are deterministic, so we will not need to explicitly analyse or discuss their running time.}. So, randomized logspace computation corresponds to algorithms running in $O(\log n)$ space \textit{and} at most poly$(n)$ time on inputs of size poly$(n)$.

This means that if the $\PCC$ is true, no logarithmic space (deterministic or randomized) algorithm can solve the planted clique detection (or recovery) problems for $k=O(n^{1/2-\delta})$. If we can show logarithmic space algorithms exist above the polynomial time threshold $k=\Omega(\sqrt{n})$, we will have shown that the statistical-computational gap holds even for space complexity.

\textbf{Detection:}

For detection, essentially the same straightforward algorithms that have been designed for time efficiency can also be implemented space efficiently. For clique sizes above the information theoretic threshold $k \geq (2+\epsilon) \log n$, the same `exhaustively search over sets of $\Theta(\log n)$ vertices' idea that gives a quasi-polynomial $n^{O(\log n)}$ time algorithm also gives a $O(\log^2 n)$ space algorithm\footnote{Since the best known time complexity for this problem is $n^{O(\log n)}$, we do not expect to solve this problem in $o(\log^2 n)$ bits of space}. For large cliques above the polynomial time threshold $k=\Omega(\sqrt{n})$, the folklore `sum test' or `edge counting' algorithm (see for example Section 1.5 of \cite{lugosi2017lectures}) can be implemented in $O(\log n)$ bits of space. We elaborate more on these algorithms in Section~\ref{sec:techniques}, but for now it suffices to observe that this means a statistical-computational gap holds for planted clique detection at $k=\Theta(\sqrt{n})$ in terms of space complexity if it holds for time complexity.

\textbf{Recovery:}

But what about planted clique recovery? Before we go any further, we should clarify what we mean by a small space algorithm for planted clique recovery. The size of the output is $k \log n$ bits, which could be much larger than the space we are allowing the algorithm. However, the space bound applies only to the working-space of the algorithm, and the output is written on a \textit{write-only} area which does not count towards the space bound. This is standard in the space complexity literature, so we can write-to-output very large answers. See, for example, Section 14.1 of \cite{wigderson2019mathematics}.

Just like for detection, simple pre-existing ideas can easily be used to obtain a $O(\log^2 n)$ space algorithm for recovering planted cliques above the information theoretic threshold, thus matching the detection space complexity in this range of parameters. We provide more details in Section~\ref{sec:techniques}. Also like for detection, we do not expect a $O(\log n)$ space algorithm in this regime because of the $\PCC$ and the relation between space and time complexity.

If we can design a $O(\log n)$ space algorithm that recovers large planted cliques $k = \Omega(\sqrt{n})$, we will have shown two things:
\begin{itemize}
	\item If the conjectured statistical-computational gap at $k=\Theta(\sqrt{n})$ holds for the time complexity of the planted clique recovery problem, it also holds for space complexity.
	\item Assuming the above statistical-computational gap holds, the coarse-grained\footnote{By coarse-grained time complexity we mean that we do not distinguish between different poly$(n)$ running times. If we were looking at a more fine-grained picture, a gap does emerge between detection and recovery. \cite{mardia2020finding} showed that for $k=\omega(n^{2/3})$, planted clique detection can be solved in $o(n)$ time. However, by results of \cite{racz2019finding} we know that any recovery algorithm must require $\Omega(n)$ time. For space complexity, `coarse-grained' means we do not distinguish between any two $O(\log n)$ space algorithms. Observe that even if there is a $O(\log n)$ space algorithm that recovers cliques of size $k=\Omega(\sqrt{n})$, the fine-grained space complexity of detection and recovery could, in principle, be different. This would be the case if there exists a $o(\log n)$ space algorithm for detection but not for recovery.} computational complexity of planted clique detection and recovery are indeed the same, no matter the notion of complexity we use - time or space.
\end{itemize}

\begin{enumerate}
\item Our first hope for such a logspace recovery algorithm is to see if any pre-existing algorithms are space efficient. However, none of the polynomial time algorithms designed for recovery above $k = \Omega(\sqrt{n})$ run in small space. They all require at least poly($n$) bits of space, and in Section~\ref{sec:techniques} we discuss, for each of them, why it seems hard to implement them in $O(\log n)$ bits of space.

Of course, the `degree-counting' polynomial time recovery algorithm for large cliques of size $k=\Omega(\sqrt{n \log n})$ from \cite{kuvcera1995expected} \textit{can} easily be implemented in $O(\log n)$ space. This matches the space complexity for detection in this parameter range. For such large cliques, a simple threshold separates the degrees of non-clique vertices and clique vertices, so membership can easily be decided from a vertex's degree. A space efficient implementation exists because we can easily count the degree of a vertex (which takes $O(\log n)$ bits to store) and iterate over all vertices in logarithmic space, re-using the counter used to store the degree across vertices. However, this idea does not work for $\Omega(\sqrt{n}) = k = o(\sqrt{n \log n})$, and it is this parameter range in which most algorithmic work for the planted clique problem has been done in the past two decades. If we want to show that the statistical-computational phenomena that hold for time complexity also hold for space complexity, we will need to focus on these parameters.

\item Our next hope is to recall that the lack of a detection-recovery gap in the time complexity of the planted clique problem is not merely an algorithmic coincidence. Section 4.3.3 of \cite{alon2007testing} shows a black box way to convert a planted clique detection algorithm into a recovery algorithm. The key idea is that if a vertex $v$ is in the clique, the subgraph induced on the vertex set that \textit{does not} contain $v$ or its neighbours is distributed as an Erd\H{o}s-R\'enyi graph. But, if $v$ is not in the clique, this induced subgraph is distributed as a planted clique graph. Then we can simply run the detection algorithm to decide if $v$ is in the planted clique or not\footnote{Of course, such a reduction has a built in $O(n)$ factor time overhead for the recovery algorithm above the detection algorithm.}. If we could use the edge counting detection algorithm and implement this reduction between recovery and detection in small space, then it seems we would be done. What is more, such a reduction \textit{can} be implemented in small space!\footnote{To count the number of edges induced in such a manner by a vertex $v$, we can simply iterate over all pairs $u,w$ of vertices in the original graph. We increment the counter only if the edge $(u,w)$ exists \textit{and} neither of the vertices $u,w$ have an edge to $v$.}

However, there is a slight issue. The statistical success of the reduction in \cite{alon2007testing} requires the failure probability of the detection algorithm to be at most $o(\frac{1}{n})$. This is because we need to repeat the detection algorithm $n$ times, once for each vertex in the original graph, and thus need to take a union bound. However, as we can see from Section 1.5 in \cite{lugosi2017lectures}, the failure probability of the edge counting test is $\exp\left(\Theta\left(\frac{-k^4}{n^2}\right)\right)$. This means the failure probability is $o(\frac{1}{n})$ only for $k = \omega(\sqrt{n} \log^{\frac{1}{4}} n)$, which is not a huge improvement over the degree counting algorithm.

\end{enumerate}

Due to the discussion above, we need some new ideas to get small space recovery algorithms for planted cliques of size $k=\Omega(\sqrt{n})$. Our main result, stated informally below, is one that falls just short of our aim of a $O(\log n)$ space algorithm. For a formal statement, see Theorem~\ref{thm:main} in Section~\ref{sec:Algorithms}.

\begin{quote}
For some large enough constant $C>0$, for planted cliques of size $k \geq C\sqrt{n}$, the recovery problem $\PCR$ can be solved in $O\left(\left(\log^*{n}-\log^*{\frac{k}{\sqrt{n}}}\right) \cdot \log n\right)$ bits of space.
\end{quote}
\begin{enumerate}
\item If $k = \omega(\sqrt{n}\log^{(\ell)}n)$ for any constant integer $\ell > 0$, the space usage is indeed $O(\log n)$ bits, which was our target.
\item However, if $k = C\sqrt{n}$, the space usage is $O(\log^*{n} \cdot \log n)$ bits, which is just shy of what we were aiming for.
\end{enumerate}

Our result suggests that there does exist an $O(\log n)$ space algorithm to recover cliques of size $k =\Omega(\sqrt{n})$, since we come very close to achieving such parameters. We fail to answer our titular question, but only just. We provide strong evidence that the answer is `yes', and the statistical-computational phenomena that (conjecturally) hold for planted clique time complexity also (conjecturally) hold for space complexity. We have thus initiated the study of high dimensional statistical problems in terms of their space complexity.

As we see in Section~\ref{sec:related_work}, a long line of work on restricted models of computation has been used to show hardness of the planted clique problem. On the other hand, this work (like \cite{mardia2020finding}) studies a restricted model of computation with the primary aim of making algorithmic progress and further pushing down the complexity of successful planted clique algorithms.

\textbf{Open Problem:}
Is there a logspace algorithm that recovers planted cliques of size $k = \Omega(\sqrt{n})$ reliably, or is there a (tiny) detection-recovery gap in the space complexity of the planted clique problem?

\subsection{Related Work}\label{sec:related_work}

\textbf{Planted Clique Hardness:} It is widely believed that polynomial time algorithms can only detect or recover the planted clique for clique sizes above $k=\Omega(\sqrt{n})$. One piece of evidence for this belief is the long line of algorithmic progress using a variety of techniques that has been unable to break this barrier \cite{kuvcera1995expected,alon1998finding,feige2000finding,feige2010finding,ames2011nuclear,dekel2014finding,chen2014statistical,deshpande2015,hajek2015computational,mardia2020finding}. The other piece of evidence comes from studying restricted but powerful classes of algorithms. \cite{jerrum1992large} showed that a natural Markov chain based technique requires more than polynomial time below this threshold. Similar hardness results (for the planted clique problem or its variants) have been shown for statistical query algorithms \cite{feldman2017statistical}, circuit classes \cite{rossman2008constant,rossman2010monotone}, the Lov{\'a}sz--Schrijver hierarchy \cite{feige2003probable}, and the sum-of-squares hierarchy \cite{meka2015sum,deshpande2015improved,hopkins2018integrality,barak2019nearly}. Further evidence comes from the low-degree-likelihood method \cite{hopkins2017bayesian,hopkins2017power,hopkins2018statistical,kunisky2019notes} and through concepts from statistical physics \cite{gamarnik2019landscape}.

\textbf{Statistical-Computational Gaps:} Statistical-computational gaps are not unique to the planted clique problem, and are found in problems involving community detection / recovery \cite{decelle2011asymptotic,massoulie2014community,mossel2015consistency,abbe2016achieving}, sparse PCA \cite{pmlr-v30-Berthet13,lesieur2015phase}, tensor PCA \cite{richard2014statistical,hopkins2017power}, random CSPs \cite{achlioptas2008algorithmic,kothari2017sum}, and robust sparse estimation \cite{li2017robust,balakrishnan2017computationally}. However, the planted clique problem is special in that its hardness (or that of its close variants) can be used to show hardness and statistical-computational gaps for a variety of other problems. Such reductions can be seen in \cite{pmlr-v30-Berthet13,alon2007testing,brennan2018reducibility,brennan2020reducibility}. See \cite{brennan2020reducibility} for a more comprehensive list of examples. To the best of our knowledge, most of these reductions use randomness quite heavily, so it is unclear if such connections can also be made using only logarithmic space reductions. It would be interesting to do so since this would tie these problems together even more tightly, and would show that planted clique is a central problem in average-case complexity not just for time but also space.

\textbf{Detection-Recovery Gaps:} As we have mentioned, the statistical-computational gap in the planted clique problem appears at $k=\Theta(\sqrt{n})$ for both the detection and recovery variants. This means there is no detection-recovery gap in time complexity, and our work is trying to show that no such gap exists for space complexity either. To understand that the non-existence of this gap is not a foregone conclusion, we note that for several other problems, detection-recovery gaps \textit{do} exist. For example, for communities in the stochastic block model \cite{abbe2017community}, or planted submatrix problems \cite{hajek2015computational,chen2014statistical}. Moreover, the (non-)existence of a detection-recovery gap is not an inconsequential detail. Since the planted clique problem does not display such a gap, it is not straightforward to use it as a starting point to show detection-recovery gaps for other problems. \cite{brennan2020reducibility} overcomes this issue for semirandom community recovery by starting from a variant of the planted clique problem, and \cite{schramm2020computational} develops a low-degree-likelihood ratio technique tailored to recovery tasks to get around this problem.

\subsection{Notation and Problem Definition}\label{sec:notation}

\paragraph{Notation:}
We will use standard big $O$ notation ($O, \Theta, \Omega$). An edge between vertices $u,v$ is denoted $(u,v)$. We let $\Binnhalf$ denote a Binomial random variable with parameters $\left(n,\frac{1}{2}\right)$. Similarly, {\sf{Bern}}($p$)
denotes a Bernoulli random variable that is $1$ with probability $p$ and $0$
otherwise. Unless stated otherwise, all logarithms are taken base $2$. For a vertex $v$ in graph
$G = ([n], E)$, we will denote its degree by $\deg(v)$. Throughout this work we identify the vertex set of the graph with the set $[n] := \{1,2,...,n\}$. We will also crucially utilise the natural ordering this confers on the names of the vertices.

We also define the so-called binary iterated logarithm $\log^{*}n$. \[
\log^{*}n =
\begin{cases}
0 & \text{if $n \leq 1$} \\
1 + \log^{*}(\log n) & \text{if $n>1$}
\end{cases}
\]

Below are formal definitions of the graphs ensembles we use and the planted clique problem.\\

\begin{definition}[Erd\H{o}s-R\'enyi graph distribution: $\G(n, \frac{1}{2})$]\ \\
	\label{def:er}
	Let $G = ([n],E)$ be a graph with vertex set of size $n$. The edge set $E$ is created by including each possible edge independently with probability $\frac{1}{2}$. The distribution on graphs thus formed is denoted $\G(n, \frac{1}{2})$.\\
\end{definition}
\begin{definition}[Planted Clique graph distribution: $\G(n, \frac{1}{2},k)$]\ \\
	\label{def:pc}
	Let $G = ([n],E)$ be a graph with vertex set of size $n$. 
	Moreover, let $K \subset [n]$ be a set of size $k$ chosen uniformly at random from all ${n \choose k}$ subsets of size $k$. For all distinct pairs of vertices $u,v \in K$, we add the edge $(u,v)$ to $E$. 
	For all remaining distinct pairs of vertices $u,v$, 
	we add the edge $(u,v)$ to $E$ independently with probability $\frac{1}{2}$. The distribution on graphs thus formed is denoted $\G(n, \frac{1}{2},k)$.\\
\end{definition}

\begin{definition}[Planted Clique Detection Problem: ${\sf{PC_D}}(n,k)$]\label{def:pc-detection}\ \\
	This is the following hypothesis testing problem.
	\begin{align*}
	\label{eq:hypotheses}
	{\sf H}_0: G \sim \G(n, \frac{1}{2}) \hspace{0.1cm} \text{ and } \hspace{0.1cm}
	{\sf H}_1: G\sim
	\G(n, \frac{1}{2},k).
	\end{align*}
	Give an algorithm $\A$ that takes as input the graph $G$ and outputs either $0$ or $1$ so that \[\Pr(\A(G)=0|{\sf H}_0) + \Pr(\A(G)=1|{\sf H}_1) \geq 4/3. \]
\end{definition}

\begin{definition}[Planted Clique Recovery Problem: ${\sf{PC_R}}(n,k)$]\ \\
	\label{def:pc-recovery}
	Given an instance of $G\sim
	\G(n, \frac{1}{2},k)$, recover the planted clique $K$ with probability at least $2/3$.
\end{definition}

\subsection{Our Techniques}\label{sec:techniques}
Our space efficient recovery algorithm will depend on the ability to take a small subset of the planted clique and expand it to recover the entire clique. We first discuss such a subroutine, and then talk about our main result, the $O\left(\left(\log^*{n}-\log^*{\frac{k}{\sqrt{n}}}\right) \cdot \log n\right)$ space algorithm for planted clique recovery for large cliques of size $k=\Omega(\sqrt{n})$. We do this by first studying polynomial time algorithms that work in this regime, discussing why they take polynomial amounts of space to implement, and then providing the high level ideas of our algorithm. After this, we end with some more details on the straightforward $O(\log^2 n)$ space implementations of the known quasi-polynomial time algorithms for clique detection and recovery above the information theoretic threshold.

\textbf{Small space clique completion:}

Several polynomial time recovery algorithms use clique completion / clean-up subroutines to find the entire planted clique after finding just a large enough (possibly noisy) subset of it \cite{alon1998finding,feige2010finding,dekel2014finding,deshpande2015,mardia2020finding}. However, none of these seem amenable to space efficient implementation, so we create a simple completion algorithm of our own.

We assume we have an algorithm that implicitly maps any planted clique graph to a specific large enough subset of the vertices of the planted clique, which we call $S_C$. If given any vertex as input, this algorithm can answer whether this vertex is in $S_C$ or not using $s(n)$ bits of space. This is what we mean by `having access to' a subset of the clique that we can now complete. Consider the set $\widetilde{V}$ of those vertices which are connected to every vertex in $S_C$. It is easy to show that this new set $\widetilde{V}$ contains the entire planted clique and very few non-clique vertices (see Lemma~\ref{lem:false=pos}). As a result, the number of edges to $\widetilde{V}$ from a clique vertex is far larger than that of a non-clique vertex, and a simple logspace computable threshold can distinguish between the two cases. We show in Algorithm~\ref{alg:completion} (\textsc{Small Space Clique Completion}) and Lemma~\ref{lem:completion} that we can use this to decide if a given vertex is in the planted clique or not using $O(\log n + s(n))$ bits of space. Then we simply loop over all vertices with a further $O(\log n)$ bits of space and thus have a planted clique recovery algorithm.

\textbf{Recovery for $k=\Omega(\sqrt{n})$:}

We first take a look at existing polynomial time algorithms for $k=\Omega(\sqrt{n})$ to see why they all require poly$(n)$ bits of space, and to see if they have good ideas that we can build on to get small space algorithms.
\begin{enumerate}
	\item Spectral algorithms: The spectral algorithm of \cite{alon1998finding}, which was the first polynomial time algorithm to recover planted cliques of size $k=\Omega(\sqrt{n})$, requires access to an $n$-dimensional eigenvector. Even just storing this takes poly($n$) bits of space, and it is unclear how to space-efficiently compute only bits and pieces of this eigenvector. Perhaps the most promising avenue for a space efficient spectral algorithm would be to use the spectral detection test (based on the second eigenvalue of the adjacency matrix) with the reduction between recovery and detection from \cite{alon2007testing}. The spectral detection test has a much smaller failure probability than $o(\frac{1}{n})$, so if we can implement it in small space, this approach would actually work. However, it is not at all clear how to compute the second eigenvalue of the adjacency matrix to desired accuracy in $O(\log n)$ space. \cite{doron2015problem,doron2017approximating} study the problem of approximating eigenvalues of an undirected graph in logarithmic space and we might hope to use their algorithms to solve our problem. However, these algorithms, which are randomized and run in logarithmic space, can only approximate the normalized eigenvalues to within constant accuracy. We require inverse polynomial accuracy to use the spectral detection test.
	\item Optimization / SDP algorithms: Several optimization theoretic algorithms involving semidefinite programs have been designed that solve the planted clique recovery problem for $k=\Omega(\sqrt{n})$ \cite{feige2000finding,ames2011nuclear,chen2014statistical,hajek2015computational}. However, we do not expect to have a general-purpose logarithmic space algorithm for semidefinite programs. The works \cite{dobkin1979linear,serna1991approximating} show that even (approximately) solving linear programs, which are a special case of semidefinite programs, is logspace complete for P. This means that if we had a logspace algorithm for semidefinite programs, every problem with a polynomial time algorithm could also be solved in logarithmic space. Such a proposition is believed to be untrue \cite[Conjecture 14.8]{wigderson2019mathematics}.
	\item (Nearly) Linear time algorithms: 
	\begin{enumerate}
		\item The algorithm of \cite{feige2010finding} maintains a subset of `plausible clique vertices' and reduces the size of this subset by $1$ in every round. As a result, it needs to maintain a polynomially large subset for most of the time it runs. There also does not seem to be a clever way to compress this set, since it depends crucially on the edge structure of the graph.
		\item The message passing algorithm of \cite{deshpande2015} is iterative and produces a new dense $n \times n$ matrix at every iteration, which can not be done in logarithmic space. It is plausible that a more space efficient recursive algorithm that recomputes messages as needed exists. But, since the algorithm requires $\Theta(\log n)$ sequential iterations / recursive calls, and we will need $\Omega(\log n)$ bits of space for each level of recursion, we do not expect this space usage to be $o(\log^2 n)$ bits. Since this does not improve the space usage over the simple algorithm that works above the information theoretic threshold, we do not pursue this idea further.
		\item The algorithm of \cite{dekel2014finding}, like \cite{feige2010finding}, maintains a sequence of shrinking subsets of vertices where the ratio between the number of clique and non-clique vertices improves in every round. Further, these subsets are polynomial sized \textit{and} random. Since the pruning of the set depends on randomness from the algorithm, any clever space efficient implementation that re-uses space would need to store the random coins it tosses, defeating the purpose of a space efficient implementation. However, the key idea behind this algorithm \textit{can} be de-randomized, and this is the first observation that forms the basis of our $O\left(\left(\log^*{n}-\log^*{\frac{k}{\sqrt{n}}}\right) \cdot \log n\right)$ space algorithm.
	\end{enumerate}
\end{enumerate}

We briefly explain the technique of \cite{dekel2014finding} in more detail, but using the notation of this work rather than that of \cite{dekel2014finding}. Their algorithm runs for $T$ rounds and maintains a sequence of vertex subsets $\{N_t,V_t\}_{ 1 \leq t \leq T}$. $N_1=V_1$ is essentially the entire vertex set $[n]$, and then each vertex of $V_{t-1}$ is included in $N_t$ iid with some probability and each vertex of $N_{t}$ is added to $V_t$ by cleverly using information from the edge structure of the input graph. This results in the ratio of clique vertices to non-clique vertices in $V_t$ increasing by a constant factor in every round. $T$ is then chosen large enough so that $V_T$ is entirely a subset of the planted clique. The entire clique is now output using a clique completion subroutine. 

Since the subsets $N_t$ described above depend so heavily on the randomness of the algorithm as well as the edge structure of the input graph, this algorithm can not be implemented in less that poly($n$) space. On the other hand, we have already noted that creating a space efficient clique completion algorithm can be done, and we have done so in Lemma~\ref{lem:completion} with Algorithm~\ref{alg:completion}. So we now focus on trying to modify the first part of the algorithm to something that can be implemented space efficiently. Our challenge is to concisely represent the sets $N_t$ (and by extension, $V_t$).

Our observation is that the clever filtering of \cite{dekel2014finding} does not depend crucially on the set $N_t$ being a subset of $V_{t-1}$ (which is what makes it depend on the edge structure of the graph). Nor does it depend on the set being random. The only thing we really need is that the proportion of clique to non-clique vertices in $N_t$ is not too small, and that we can easily iterate over all the vertices in any set $N_t$. This gives us the freedom we require to design the sets $N_t$ to be concisely representable, and we use our computer's representation of the vertex set to our advantage. For our computer, the names of the $n$ vertices of the graph are $\log n$ bit integers, and we can use the fact that integers have a natural ordering as well as the fact that simple arithmetic can easily be done in $O(\log n)$ bits of space.

We first set up some notation. The quantities we define will be functions of $n,k$, and the graph $G \sim \G(n, \frac{1}{2},k)$ although our notation will not explicitly denote this. The value of $n,k$ and the graph will always be clear from context. Recall that $[n]$ is the vertex set of a graph $G \sim \G(n, \frac{1}{2},k)$ with $n$ vertices and a planted clique called $K$ of size $k$ and a set of edges called $E$.

\begin{itemize}
	\item Define $n_0$ as the smallest integer that is a power of $2$ and is at least $n/2$. This means $n/2 \leq n_0 < n$. Define $k_0 := k \frac{n_0}{n}$
	\item For any integer $0 < t < \log {n_0}$, let $n_t := \frac{n_0}{2^t}$, $k_t := \frac{k_0}{2^t}$. Note that $n_t$ is always an integer.
	\item We can now define the subsets $N_t$ of the vertex set $[n]$ that will be of particular interest in our filtering algorithm. Let $N_{t} := [n_{t-1}] \setminus [n_t]$, and note that the $N_t$'s are all disjoint sets. Clearly, $|N_t| = n_t$. See Figure~\ref{fig:ntexample} for an example.
\end{itemize}

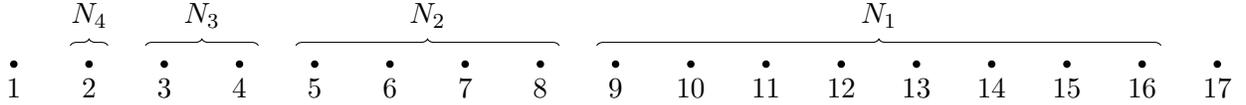
\begin{figure}
\begin{centering}

\begin{tikzpicture}
\tikzset{every node}=[draw,shape=circle]

\draw[xshift=-1cm] (1,0) 
node[circle,fill,inner sep=1pt,label=below:$1$](b){};
\draw[xshift=-1cm] (2,0) 
node[circle,fill,inner sep=1pt,label=below:$2$](b){};
\draw[xshift=-1cm] (3,0) 
node[circle,fill,inner sep=1pt,label=below:$3$](b){};
\draw[xshift=-1cm] (4,0) 
node[circle,fill,inner sep=1pt,label=below:$4$](b){};
\draw[xshift=-1cm] (5,0) 
node[circle,fill,inner sep=1pt,label=below:$5$](b){};
\draw[xshift=-1cm] (6,0) 
node[circle,fill,inner sep=1pt,label=below:$6$](b){};
\draw[xshift=-1cm] (7,0) 
node[circle,fill,inner sep=1pt,label=below:$7$](b){};
\draw[xshift=-1cm] (8,0) 
node[circle,fill,inner sep=1pt,label=below:$8$](b){};
\draw[xshift=-1cm] (9,0) 
node[circle,fill,inner sep=1pt,label=below:$9$](b){};
\draw[xshift=-1cm] (10,0) 
node[circle,fill,inner sep=1pt,label=below:$10$](b){};
\draw[xshift=-1cm] (11,0) 
node[circle,fill,inner sep=1pt,label=below:$11$](b){};
\draw[xshift=-1cm] (12,0) 
node[circle,fill,inner sep=1pt,label=below:$12$](b){};
\draw[xshift=-1cm] (13,0) 
node[circle,fill,inner sep=1pt,label=below:$13$](b){};
\draw[xshift=-1cm] (14,0) 
node[circle,fill,inner sep=1pt,label=below:$14$](b){};
\draw[xshift=-1cm] (15,0) 
node[circle,fill,inner sep=1pt,label=below:$15$](b){};
\draw[xshift=-1cm] (16,0) 
node[circle,fill,inner sep=1pt,label=below:$16$](b){};
\draw[xshift=-1cm] (17,0) 
node[circle,fill,inner sep=1pt,label=below:$17$](b){};

\draw[decoration={brace,raise=7pt},decorate](7.75,0) -- node[above=10pt] {$N_1$}(15.25,0);
\draw[decoration={brace,raise=7pt},decorate](3.75,0) -- node[above=10pt] {$N_2$}(7.25,0);
\draw[decoration={brace,raise=7pt},decorate](1.75,0) -- node[above=10pt] {$N_3$}(3.25,0);
\draw[decoration={brace,raise=7pt},decorate](0.75,0) -- node[above=10pt] {$N_4$}(1.25,0);

\end{tikzpicture}
\end{centering}

\caption{An example of our sets $N_t$ for $n=17$.}
\label{fig:ntexample}
\end{figure}

It is easy to observe that given $n,t$ we can iterate over the vertex set $N_t$ in $O(\log n)$ space, which is exactly what we wanted. In the analysis of Lemma~\ref{lem:struct-filt}, we will also show that the ratio of clique vertices to non-clique vertices in any $N_t$ is roughly the same as in the whole graph, which is not too small. Now we must implement the rest of the ideas in \cite{dekel2014finding}, the ones that actually use the input graph to find the clique.

After setting $V_1 = N_1$, the filtering step of \cite{dekel2014finding} fixes a threshold and adds a vertex in $N_t$ to the set $V_t$ if and only if that vertex has more edges to $V_{t-1}$ \footnote{Technically, \cite{dekel2014finding} counts the number of edges to $V_{t-1} \setminus N_t$, but in our construction we will have $V_{t-1} \setminus N_t = V_{t-1}$.} than the set threshold. Since clique vertices in $N_t$ are likely to have a higher number of edges to $V_{t-1}$ than non-clique vertices, the former are more likely to appear in $V_t$ and the latter are more likely to be filtered out. This is how the ratio of clique to non-clique vertices in $V_t$ gradually increases with $t$. If we had an algorithm to check membership in $V_{t-1}$ that uses $s_{t-1}(n)$ bits of space, we could design an algorithm to check for membership in $V_t$ that uses $O(\log n) + s_{t-1}(n)$ bits of space. To see this, suppose we have a vertex $v \in N_t$ and we want to decide if it is in $V_{t}$. We can simply iterate over the set $N_{t-1}$, and for each vertex $u \in N_{t-1}$, check if it is also in $V_{t-1}$ using our assumed algorithm. We can also maintain a $O(\log n)$ bit counter to count the number of edges from $v$ to all $u$ that are in $V_{t-1}$. Since we can re-use the $s_{t-1}(n)$ bits of space to check membership in $V_{t-1}$ across different $u$, the whole things can be done in $O(\log n) + s_{t-1}(n)$ bits of space. By induction, this means we can check for membership in $V_T$ using $O(T \cdot \log n)$ bits of space. We provide a formal algorithm and proof of such a claim in Lemma~\ref{lem:smallspacefilterimpl} using Algorithm~\ref{alg:filter}.

Overall, this promises to give a $O(T \cdot \log n)$ space algorithm. What can we set $T$ to be? Unfortunately, the algorithm of \cite{dekel2014finding} uses $T = \Theta(\log n)$ rounds, since it only gets a constant factor improvement in the ratio of clique to non-clique vertices in going from $V_{t-1}$ to $V_t$. This gives a $O(\log^2 n)$ space algorithm, which is not an improvement over the simple algorithm that works above the information theoretic threshold.

Our key idea, inspired by \cite{mardia2020finding}, is to implement a better filtering step that gets more than a constant factor of improvement in each round. The filtering / thresholding of \cite{dekel2014finding} does not utilise the size of the planted clique $k$ at all, other than the fact that it is $\Omega(\sqrt{n})$. On the other hand, \cite{mardia2020finding} uses knowledge of $k$ to design a single round filtering algorithm that recovers the planted clique for clique sizes $\omega(\sqrt{n \log \log n}) = k = o(\sqrt{n \log n})$ in sublinear time. By appropriately implementing this idea in our context for multiple rounds, we can utilize knowledge of the number of clique vertices in $V_{t-1}$, $|V_{t-1} \cap K|$, to make sure that in going from $V_{t-1}$ to $V_t$ the following happens. The number of clique vertices decreases by at most a constant factor, while the number of non-clique vertices decreases by at least a factor of $\exp\left(\Theta\left(\frac{|V_{t-1} \cap K|^2}{|V_{t-1}|}\right)\right)$, which is $\exp\left(\Theta\left(\frac{k^2}{n}\right)\right)$ for $t=2$. For $k = \Theta(\sqrt{n})$, this is still a constant factor, but for larger $k$, this is much better than a constant factor improvement.

To use this idea, our algorithm needs to know $|V_{t-1} \cap K|$, which it does not. However, we do have high probability lower bounds on the size $|V_{t-1} \cap K|$. We design our thresholds using these estimates, and our analysis in Lemma~\ref{lem:struct-filt} shows that this suffices to get the benefits of this better filter. Let us now define the sets $V_t$ for our algorithm, thus specifying the filtering threshold. We proceed inductively.
\begin{itemize}
	\item $V_1 := N_1$
	\item For any integer $t > 1$, $V_t$ is a subset of $N_t$ of vertices which have `large' $V_{t-1}$-degree. Quantitatively, $V_t := \{v \in N_t | \sum\limits_{u \in V_{t-1}} \mathbb{1}_{(u,v) \in E} \geq \frac{|V_{t-1}|}{2} + k_{t+2} - 2\sqrt{|V_{t-1}|}  \}$.
\end{itemize}

It is this carefully chosen threshold sequence which, unlike in \cite{dekel2014finding}, varies with $t$ and uses the value of $k$ that allows us to improve on the $O(\log^2 n)$ space bound. In Lemma~\ref{lem:struct-filt} we will show that $V_T$, as defined above, is with high probability a subset of the planted clique if $T$ is large enough. We can implement an algorithm to check membership in $V_T$ in $O(T \cdot \log n)$ bits of space as discussed above (and formalized in Lemma~\ref{lem:smallspacefilterimpl}). Moreover, we get the benefits of a very quickly accelerating improvement in the ratio of clique to non-clique vertices from $V_{t-1}$ to $V_t$. From \cite{mardia2020finding} we know that one round of such a filter improves the ratio by a factor of $\exp\left(\Theta\left(\frac{k^2}{n}\right)\right)$, and the analysis of our filtering in Lemma~\ref{lem:struct-filt} shows that after $t$ rounds of such filtering, the ratio improves by what is essentially a tower of exponentials of height $\frac{t}{2}$, i.e. $\exp\left(\exp\left(...\exp\left(\Theta\left(\frac{k}{\sqrt{n}}\right)\right)\right)\right)$. This is why we are able to take $T = O\left(\log^{*} n - \log^{*} \frac{k}{\sqrt{n}}\right)$ (Lemma~\ref{lem:struct-filt}). This gives us our main result, an algorithm that can recover planted cliques of size $k \geq C\sqrt{n}$ in $O\left(\left(\log^*{n}-\log^*{\frac{k}{\sqrt{n}}}\right) \cdot \log n\right)$ bits of space. The formal statement and proof can be found in Theorem~\ref{thm:main}.

\textbf{Detection:}

\begin{enumerate}
	\item It is well known (see \cite{bollobas1976cliques} or Lemma~\ref{lem:gnpcliquesize}) that for any positive constant $\epsilon >0$, the probability that an Erd\H{o}s-R\'enyi $\G(n, \frac{1}{2})$ graph has a clique of size at least $(2+\epsilon) \log n$ goes to $0$. Meanwhile, if $k \geq (2+\epsilon) \log n$, then by definition a planted clique graph $\G(n, \frac{1}{2},k)$ has a clique of size $(2+\epsilon) \log n$. The existence of a clique of this size is a well-known and simple detection test for $\PCD$ (see, for example, Proposition 1.3 \cite{lugosi2017lectures}). Moreover, such a test only needs to iterate over all vertex subsets of size $(2+\epsilon) \log n$, which can be done by maintaining a $\log n$ bit name/number for each of the $(2+\epsilon) \log n$ vertices and looping over all possibilities. For a given possible clique, the algorithm needs to check if all ${(2+\epsilon) \log n \choose 2}$ edges exist. This can be done by looping over all these edges with 2 more $O(\log\log n)$ bit counters. Overall, this implementation requires $O(\log^2 n)$ bits of space.
	
	\item The simple `sum test' or `edge counting' algorithm that is well-known to work for large planted clique $k=\Omega(\sqrt{n})$ detection (see for example Section 1.5 of \cite{lugosi2017lectures}) can easily be implemented in $O(\log n)$ space. The planted graph has significantly more edges than the graph without a clique, so simply counting the number of edges in the input graph and using a threshold test gives a successful detection algorithm. The algorithm only needs to maintain the edge count, which is a number between $1$ and $n^2$ (which can be done with $O(\log n)$ bits), and it can also easily iterate over all distinct vertex pairs in $O(\log n)$ bits of space. Lastly, the algorithm also needs to compute the threshold (from \cite{lugosi2017lectures}, we can use the threshold $\frac{{n \choose 2}}{2}+\frac{{k \choose 2}}{4}$), which can easily be computed from the input (which contains $n,k$) in logarithmic space. This means that for planted clique detection, assuming we have a time complexity based statistical-computational gap, we also have a space complexity based statistical-computational gap.
\end{enumerate}

\textbf{Recovery above the information theoretic threshold:}

For cliques of size $(2+\epsilon) \log n \leq k = O(\log n)$, with high probability the planted clique is the unique largest clique in $\G(n, \frac{1}{2}, k)$ \cite[Theorem 1.7]{lugosi2017lectures}. This means that an algorithm that loops over all possible vertex subsets of size $k$ can find and output the entire planted clique. To do this it only need to maintain $k$ names of vertices (which takes $O(k\log n)$ bits of space) and 2 counters of $O(\log k)$ bits of space to check if a given set of $k$ vertices form a clique. Overall, this implementation needs $O(\log^2 n)$ bits of space.

A simple application of the reduction between detection and recovery from \cite{alon2007testing} combined with the $O(\log^2 n)$ space detection algorithm for clique sizes above the information theoretic threshold $k  \geq (2+\epsilon) \log n$ also gives a $O(\log^2 n)$ space recovery algorithm for $k = \omega(\log n)$. We provide a formal statement and proof in Lemma~\ref{lem:log2n-recovery}.


\section{Algorithms}
\label{sec:Algorithms}
We now prove our main results after formalizing our model of computation in Section~\ref{sec:model}. In Section~\ref{sec:cliquecompl} we give a space efficient algorithm for clique completion. In Section~\ref{sec:algmain} we prove our $O\left(\left(\log^*{n}-\log^*{\frac{k}{\sqrt{n}}}\right) \cdot \log n\right)$ space recovery algorithm for clique sizes above the polynomial time threshold $k=\Omega(\sqrt{n})$.

\subsection{Model of Computation}
\label{sec:model}

We use a standard notion of deterministic space bounded computation. See, for example, \cite[Section 14.1]{wigderson2019mathematics}. For a $s(n)$-space algorithm, the input is a read-only version of the $n \times n$ adjacency matrix of the graph as well as the clique size $k$. Every entry in the matrix as well as the value of $k$ is stored in its own register. The algorithm has access to $s(n)$ bits of working space, and the output is write-only (and possibly much larger than $s(n)$). The last fact allows us to solve problems whose outputs may be much larger than $s(n)$, a property we will use to solve $\PCR$.

To make our model convenient for algorithm design, we also allow random access to the input registers. In our model, we assume basic arithmetic (addition, multiplication, subtraction, division) on $O(\log n)$ bit numbers can be done in $O(\log n)$ bits of space. We also assume that the algorithm can compute or knows $n$ by accessing the adjacency matrix using $O(\log n)$ bits of space.

\begin{subsection}{Space bounded clique completion}\label{sec:cliquecompl}

The main idea behind this algorithm is discussed in Section~\ref{sec:techniques}. If we have access to a large enough subset of the clique, very few vertices that are adjacent to the entire subset (i.e `common neighbours') are not in the planted clique. Counting the edges from a vertex $v$ to this set of `common neighbours' of the known clique subset allows us to decide whether or not $v$ is in the planted clique.
\vspace{.2cm}
\begin{algorithm}[]
	\SetAlgoLined
	\KwIn{Graph $G = ([n],E) \sim \G(n, \frac{1}{2},k)$, clique size $k$, oracle $O_{S_C}$ with access to a clique set $S_C \subset K$ : $O_{S_C}(v) = 1$ if $v \in S_C$, $O_{S_C}(v) = 0$ if $v \notin S_C$ }
	\KwOut{Clique $K$}
	
	\For{$v \in [n]$}{
			Initialize $\widetilde{\text{deg}}(v) = 0$
			
			\For{$u \in [n]$}{
				Initialize 	$in\widetilde{V}(u) =$ TRUE
				
				\For{$w \in [n]$}{\If{$O_{S_C}(w)=1$ and $(w,u) \notin E$}{
						
						\hspace{15em}\smash{$\left.\rule{0pt}{3.5\baselineskip}\right\}\ \mbox{Decide if $u$ is a common neighbour}$}
						
						Set $in\widetilde{V}(u) =$ FALSE}}
				
				\If{$in\widetilde{V}(u) = \text{TRUE}$ and $(u,v) \in E$}{$\widetilde{\text{deg}}(v) = \widetilde{\text{deg}}(v) + 1$}
				
			}
			
			\hspace{19.5em}\smash{$\left.\rule{0pt}{3.5\baselineskip}\right\}\ \mbox{Use `common neighbour'-degree of $v$}$}
			
			\If{$\widetilde{\text{deg}}(v) \geq \frac{2k}{3}+3 \log k$}{\textbf{write-to-output} $v$}
	}

	\caption{\textsc{Small Space Clique Completion (SSCC)}}
	\label{alg:completion}
\end{algorithm}
\vspace{.2cm}

\begin{lemma}[Deterministic + small space clique completion]\ \\
	\label{lem:completion}
	Let $k = \omega(\log n)$, and $G\sim
	\G(n, \frac{1}{2},k) = ([n],E)$. Let $O_{S_C}$ be a deterministic algorithm that uses $s(n)$ bits of space and, except with probability at most $p(n) \leq \frac{1}{2}$ (over the randomness in $G$), has the following properties.
	\begin{enumerate}
		\item When given as input the graph $G$ and clique size $k$, it implicitly defines a subset of the planted clique vertices $S_C$ such that $S_C \subset K$ and $|S_C| \geq 2 \log n$.
		\item It does this by returning, for $v \in [n]$, $O_{S_C}(v)=1$ if and only if $v \in S_C$, and $0$ otherwise.
	\end{enumerate}

	Then for large enough $n$, \textsc{Small Space Clique Completion} (Algorithm~\ref{alg:completion}), when run on $G$ with access to the algorithm $O_{S_C}$, runs deterministically in space $O(s(n) + \log n)$ and writes to output the correct planted clique $K$ except with probability at most $p(n)+\left(\frac{1}{n}\right)^{\log k} + n\exp\left(\frac{-k}{54}\right)$ (which is over the randomness in $G$).
\end{lemma}

\begin{proof}
	
\textbf{Space usage:}
The algorithm needs to store vertices $v,u,w$ to run the for loops, each of which take $\log n$ bits of space since the size of the vertex set is $n$. The for loops can be run simply by incrementing the counter that stores the name/number of $v$, $u$, or $w$. The algorithm also needs to invoke the oracle $O_{S_C}$ which we know takes $s(n)$ bits of space. The only other variables the algorithm needs to store are $\widetilde{\text{deg}}(v)$ and $in\widetilde{V}(u)$, which take $\log n$ bits (because $\widetilde{\text{deg}}(v)$ ranges from $0$ to $n-1$) and 1 bit of space respectively. Note that the algorithm also needs to compute $\frac{2k}{3}+\log k$ which can be done upto the few bits of precision required to make the comparison in $O(\log n)$ bits of space. Hence the entire algorithm has a space requirement of $O(s(n) + \log n)$ bits. Note that both $\widetilde{\text{deg}}$ and $in\widetilde{V}$ are only ever required for one $u,v$ pair at a time, and so their space is re-used across the outer for loops. Similarly, space can be re-used for every call to the oracle.

\textbf{Correctness:}
By assumption, we know that except with probability at most $p(n)$ (over the randomness of the input graph) the oracle $O_{S_C}$ outputs $1$ only on a set $S_C$ with the following properties. $S_C \subset K$ and $|S_C| \geq 2 \log n$. We shall call this event $A_1$ and condition on it happening for the rest of this proof. Let the event $C$ denote the correctness of our algorithm, and note that we are trying to upper bound $\Prob(C^c) \leq \Prob(C^c,A_1)+\Prob(A_1^c) \leq \Prob(C^c,A_1)+p(n)$.

We need to argue that the algorithm writes to output every vertex in $K$ and no other vertices. Consider the vertex set $\widetilde{V}$ consisting of vertices that have edges to every vertex in the known clique set $S_C$. For every vertex $v$ in $[n]$, our algorithm computes the number of edges from $v$ to $\widetilde{V}$ (we call this $\widetilde{\text{deg}}(v)$). This is because an edge $(u,v)$ is counted towards $\widetilde{\text{deg}}(v)$ only if $in\widetilde{V}(u)=TRUE$, which happens only when $u \in \widetilde{V}$. The algorithm then writes $v$ to output if $\widetilde{\text{deg}}(v) \geq \frac{2k}{3} + 3 \log k$ and otherwise does nothing.

To complete our proof, we need to show two things. First, for every clique vertex $v$ in $K$, $\widetilde{\text{deg}}(v) \geq \frac{2k}{3} + 3 \log k$. This happens because the entire clique $K$ is contained in $\widetilde{V}$ once we have conditioned on $A_1$, and $k-1 \geq \frac{2k}{3} + 3 \log k$ for $k$ large enough.

Second, we need to show that for every non-clique vertex $v \in [n] \setminus K$, $\widetilde{\text{deg}}(v) < \frac{2k}{3} + 3 \log k$. To do this, we use some structural properties of the random input graph. Let $A_2$ be the event that the maximum number of clique vertices any non-clique vertex is connected to is less than $\frac{2k}{3}$, and let $A_3$ be the event that the structural facts guaranteed by Lemma~\ref{lem:false=pos} are true. If $A_2$ and $A_3$ happen, then it is clear that our algorithm behaves as desired. Hence, $\Prob(C^c,A_3,A_2,A_1)=0$. Thus, $\Prob(C^c,A_1) \leq \Prob(C^c,A_3,A_2,A_1) + \Prob(A_3^c,A_1) + \Prob(A_2^c,A_1) \leq \Prob(A_3^c) + \Prob(A_2^c)$. Lemma~\ref{lem:false=pos} shows $\Prob(A_3^c) \leq \left(\frac{1}{n}\right)^{\log k}$ and Lemma~\ref{lem:2kby3} shows $\Prob(A_2^c) \leq n \exp\left(\frac{-k}{54}\right)$, which completes the proof.
\end{proof}

\end{subsection}

\begin{subsection}{Finding a clique subset in small space}\label{sec:algmain}

We recall some notation defined in Section~\ref{sec:techniques}.
\begin{itemize}
\item Define $n_0$ as the smallest integer that is a power of $2$ and is at least $n/2$. This means $n/2 \leq n_0 < n$. Define $k_0 := k \frac{n_0}{n}$
\item For any integer $0 < t < \log {n_0}$, let $n_t := \frac{n_0}{2^t}$, $k_t := \frac{k_0}{2^t}$. Note that $n_t$ is always an integer.
\item We also define some subsets of the vertex set $[n]$ that will be of particular interest in our filtering algorithm. Let $N_{t} := [n_{t-1}] \setminus [n_t]$, and note that the $N_t$'s are all disjoint sets. Clearly, $|N_t| = n_t$.

\end{itemize}

So far, we have defined vertex subsets that do not depend at all on the edge structure of the graph. Now we define some subsets that do incorporate information about such edge structure (and hence will be useful in finding the planted clique). We proceed inductively.

\begin{itemize}
\item $V_1 := N_1$
\item For any integer $t > 1$, $V_t$ is a subset of $N_t$\footnote{Hence the $V_t$'s are all disjoint for different values of $t$.} of vertices which have `large' $V_{t-1}$-degree\footnote{Defined as the number of edges from a vertex $v \in V_t$ to $V_{t-1}$.}. Quantitatively, $V_t := \{v \in N_t | \sum\limits_{u \in V_{t-1}} \mathbb{1}_{(u,v) \in E} \geq \frac{|V_{t-1}|}{2} + k_{t+2} - 2\sqrt{|V_{t-1}|}  \}$.
\end{itemize}

Our main structural lemma shows that for large enough $T$, $V_T$ is a large enough subset of the planted clique.

\begin{lemma}[Filtering lemma]\ \\
\label{lem:struct-filt}
Let $C>0$ be some large enough constant. Let $G \sim \G(n, \frac{1}{2},k)$, with $C\sqrt{n} \leq k$ and $T$ be an integer such that $2\left(\log^*{n}-\log^*{(k/\sqrt{n})}\right)+3 \leq T = O(\log^{*}n)$. Then for large enough $n$, except with probability at most $O\left(\exp\left(-n^{0.48}\right)\right)$, $V_T \subset K$ and $\omega(\log n) = \frac{k}{2^{T+3}} \leq |V_T|$.
\end{lemma}
\begin{proof}
\textbf{Step 1:}

First, we show that with high probability (over the choice of planted vertices) the number of planted vertices in each subset $N_t$ is very close to what we would expect. Fix some $1 \leq t \leq T$. By linearity of expectation, $\mathbb{E}\left[|N_t \cap K|\right] = (k/n) \times n_t = k_t = \frac{k_0}{2^t}$. Since $|N_t \cap K|$ is a hypergeometric random variable, we can use concentration inequalities for hypergeometric random variables \cite[Theorem 1]{hush2005concentration} to conclude that \[ \frac{k_0}{2^{t+1}} = \frac{k_0}{2^t} -  \frac{k_0}{2^{t+1}} \leq |N_t \cap K| \leq \frac{k_0}{2^t} +  \frac{k_0}{2^{t+1}} = \frac{3k_0}{2^{t+1}}\] except with probability at most $2\exp\left(-\frac{k}{2^{2t+10}}\right)$. Union bounding over all values of $t$ from $1$ to $T$, we see that this concentration fact is simultaneously true (which we call event $A_0$) for all such $t$ except with probability at most $\Prob(A_0^c) = O\left(\exp\left(-n^{0.49}\right)\right)$. Here we have used $T = O(\log^{*}n)$ as well as $k = \Omega(\sqrt{n})$.

\textbf{Step 2:}

We now show that (conditioned on $A_0$) with high probability, at least half the clique vertices in $N_t$ are also present in the filtered set $V_t \subset N_t$. Let $A_t$ denote the event that $V_t$ has at least $\frac{k_0}{2^{t+2}}$ clique vertices. That is, \[|V_t \cap K| \geq \frac{k_0}{2^{t+2}}.\] For the base case $\Prob(A_1^c|A_0)$ is trivially $0$ since $V_1 = N_1$.

Consider $\Prob(A_t^c|A_{t-1},A_0)$. Since $A_0,A_{t-1}$, we know that there are at least $\frac{k_0}{2^{t+1}}$ clique vertices in $N_t$ as well as $V_{t-1}$. For a given clique vertex in $N_t$, what is the probability that it is also in $V_t$? If $v \in N_t \cap K$, then $\sum\limits_{u \in V_{t-1}\setminus K} \mathbb{1}_{(u,v) \in E}$ is a ${\sf Bin}\left(|V_{t-1}|-\widetilde{k}_{t-1},\frac{1}{2}\right)$ random variable where $\widetilde{k}_{t-1} = |V_{t-1} \cap K| \geq \frac{k_0}{2^{t+1}}$. Using the Chernoff Bound (Lemma~\ref{lem:chernoff})\footnote{We can assume $|V_{t-1}|-\widetilde{k}_{t-1} > 4\sqrt{|V_{t-1}|}$ because if not, then clearly we have $p_t=1$.},
\begin{align*}
1-p_t := \Prob\left(v \notin V_t | A_{t-1},A_0\right) & = \Prob\left(\sum\limits_{u \in V_{t-1}} \mathbb{1}_{(u,v) \in E} \leq \frac{|V_{t-1}|}{2} + k_{t+2} - 2\sqrt{|V_{t-1}|} \right) \\ & = \Prob\left(\sum\limits_{u \in V_{t-1}\setminus K} \mathbb{1}_{(u,v) \in E} - \left(\frac{|V_{t-1}|-\widetilde{k}_{t-1}}{2}\right) \leq -\left(2\sqrt{|V_{t-1}|}- k_{t+2} + \frac{\widetilde{k}_{t-1}}{2} \right)\right) \\&
\leq \Prob\left(\sum\limits_{u \in V_{t-1}\setminus K} \mathbb{1}_{(u,v) \in E} - \left(\frac{|V_{t-1}|-\widetilde{k}_{t-1}}{2}\right) \leq -2\sqrt{|V_{t-1}|}\right) \\&
\leq \exp\left(\frac{-8|V_{t-1}|}{3(|V_{t-1}|-\widetilde{k}_t)}\right) \leq \exp\left(-8/3\right) \leq 0.25.
\end{align*}

Since each clique vertex in $N_t$ is added to $V_t$ independently, the total number of clique vertices in $V_t$, $\widetilde{k}_t$ is the sum of at least $\frac{k_0}{2^{t+1}}$ iid ${\sf{Bern}}(p_t)$ random variables. Using the Chernoff Bound (Lemma~\ref{lem:chernoff}), this means $|V_t \cap K| = \widetilde{k}_t \geq \frac{k_0}{2^{t+2}}$ except with probability at most $\exp\left(-c\frac{k}{2^t}\right)$ for some constant c. Hence, $\Prob(A_t^c|A_{t-1},A_0) \leq \exp\left(-c\frac{k}{2^t}\right) = O\left(\exp\left(-n^{0.49}\right)\right)$. Again, we have used $T = O(\log^{*}n)$.

We are now in a position to understand the probability that all the events $A_t$ for $0 \leq t \leq T$ happen simultaneously (which we call $A$). $\Prob(A^c) \leq \sum_{t=0}^{T} \Prob(A_t^c|A_{t-1},A_{t-2},...,A_0)$. But conditioned on the events $A_{t-1},A_0$, the event $A_t$ is indpendent of $A_1,...,A_{t-2}$. This gives $\Prob(A^c) \leq \sum_{t=0}^{T} \Prob(A_t^c|A_{t-1},A_0) = O\left(T\exp\left(-n^{0.49}\right)\right)$.

\textbf{Step 3:}

If $A$ happens, then $|V_T \cap K| \geq \frac{k_0}{2^{T+2}} \geq \frac{k}{2^{T+3}}$, which means we only need to additionally show that $V_T \subset K$ to complete the proof. To this end, we will show that the number of non-clique vertices in $V_t$ are small for all $1 \leq t \leq T$ simultaneously with high probability. Before doing so, we must set up some further notation. Define \[\frac{m_0k_2}{n_0} := m_1 := \frac{k}{\sqrt{n}} \text{ and for } t \geq 2, m_{t} := 2^{\left(\frac{m_{t-1}}{2^{t-1}}\right)}.\]

Now define $B_t$ for $t \geq 1$ as the complicated looking event that
\[|V_{t} \setminus K| \leq 
\begin{cases}
\max\left\{\frac{m_1n_{t}}{m_{t}},k_{t+2}\right\} & \text{if $\frac{m_1n_{t-1}}{m_{t-1}} \geq k_{t+1}$} \\
0 & \text{if $\frac{m_1n_{t-1}}{m_{t-1}} < k_{t+1}$}
\end{cases}
.\]
Observe that $\Prob(B_1) = 1$ because $m_1 = \frac{m_0k_2}{n_0}$ and $|V_1 \setminus K| \leq |N_1| = n_1$. We will now show that $\Prob(B_t^c|B_{t-1},A)$ is small even for $t \geq 2$. After conditioning on $B_{t-1},A$, what is the probability that a given non-clique vertex in $N_t$ is added to $V_t$? Let $v \in N_t \setminus K$, and consider $\Prob\left(v \in V_t | B_{t-1},A\right)$. Since we have conditioned on $B_{t-1}$, $|V_{t-1} \setminus K|$ is suitably small, as defined above. We can use this to upper bound $|V_{t-1}|$.

In particular, we make sure that with $k \geq C\sqrt{n}$, $C>0$ is large enough so that for large enough $n$, $k_{t+3} \leq k_{t+2} - 2\sqrt{|V_{t-1}|}$. This is equivalent to $|V_{t-1}| \leq k_{t+4}^2$ for all $t \geq 2$. Let us show that this is indeed true. Because of $A$, $|V_{t-1} \cap K| \leq 3k_t$. If $k_{t+1} > \frac{m_1n_{t-1}}{m_{t-1}}$, then $|V_{t-1}| \leq 3k_t+k_{t+1} \leq k_{t+4}^2$ for large enough $n$. If, on the other hand, $k_{t+1} \leq \frac{m_1n_{t-1}}{m_{t-1}}$, we use the fact that $m_{t-1} \geq 4^{t-2} m_1$ which we prove later in \textbf{Step 4}. $|V_{t-1}| \leq 3k_t+\frac{n_{t-1}}{4^{t-2}} \leq k_{t+4}^2$ because we have chosen $C$ large enough.

Armed with this inequality $k_{t+3} \leq k_{t+2} - 2\sqrt{|V_{t-1}|}$ and a Chernoff Bound (Lemma~\ref{lem:chernoff}), we have

\begin{align*}
q_t := \Prob\left(v \in V_t| B_{t-1},A\right) &= \Prob\left(\sum\limits_{u \in V_{t-1}} \mathbb{1}_{(u,v) \in E} \geq \frac{|V_{t-1}|}{2} + k_{t+2} - 2\sqrt{|V_{t-1}|} \right) \\
& \leq \Prob\left(\sum\limits_{u \in V_{t-1}} \mathbb{1}_{(u,v) \in E} \geq \frac{|V_{t-1}|}{2} + k_{t+3} \right) \\
& \leq \exp\left(-\frac{2k_{t+3}^2}{3|V_{t-1}|}\right).
\end{align*}

\begin{itemize}
\item \textbf{Case 1:} First we tackle the easy case. Suppose $\frac{m_1n_{t-1}}{m_{t-1}} < k_{t+1}$.

Since we have conditioned on $A,B_{t-1}$, $|V_{t-1} \cap K| \geq k_{t+1}$, which means $|V_{t-1} \setminus K| < |V_{t-1} \cap K|$. Thus $|V_{t-1}| < 2|V_{t-1} \cap K| \leq 2|N_{t-1} \cap K| \leq 3k_{t-1}$. This gives $q_t \leq \exp(-ck_t)$ for some constant $c>0$, and by a union bound over all $v \in N_t \setminus K$, we get that $|V_t \setminus K| = 0$ (which is a sufficient condition for $B_t$) except with probability at most $|N_t \setminus K|\exp(-ck_t) \leq n_t\exp(-ck_t) =  O\left(\exp\left(-n^{0.49}\right)\right)$, because $T = O(\log^*{n})$.
\item \textbf{Case 2:} Now we tackle the case $\frac{m_1n_{t-1}}{m_{t-1}} \geq \frac{k_0}{2^{t+1}}$.

Since $B_{t-1}$ and $A$ have happened, we have $|V_{t-1}\setminus K| \leq \frac{m_1n_{t-1}}{m_{t-1}}$ and $|V_{t-1} \cap K| \leq |N_{t-1} \cap K| \leq 3k_t \leq 6\frac{m_1n_{t-1}}{m_{t-1}}$, which gives $|V_{t-1}| \leq 7\frac{m_1n_{t-1}}{m_{t-1}}$. Using this with our upper bound on $q_t$, we get $q_t \leq \exp\left(-c \frac{k_{t+3}^2 m_{t-1}}{m_1n_{t-1}}\right)$ for some constant $c>0$. With $k \geq C\sqrt{n}$, let $C>0$ also be a large enough constant so that $m_1 =\frac{k}{\sqrt{n}} \geq 32$ and \[q_t \leq 2^{-\frac{m_{t-1}}{2^{t-1}}} = \frac{1}{m_t},\footnote{It is this upper bound on $q_t$ that leads to us defining $m_t$ the way we do, and thus dictates, eventually, our space complexity bound.}\] which gives $\mathbb{E}[|V_t \setminus K|] \leq \frac{n_t}{m_t}$.
Since each vertex in $N_t \setminus K$ gets added to $V_t$ independently, we can use the Chernoff Bound (Lemma~\ref{lem:chernoff}) to control $\Prob(B_t|B_{t-1},A)$. We now have the following, using the fact that $m_1$ is at least a large constant greater than $2$.

If $k_{t+2}\geq \frac{m_1n_t}{m_t}$, we have the upper bound $\Prob(B_t^c|B_{t-1},A) \leq \exp\left(-\frac{(m_1-1)k_{t+2}}{2m_1}\right) = O\left(\exp\left(-n^{0.49}\right)\right)$.

If $k_{t+2} \leq \frac{m_1n_t}{m_t}$, we have the upper bound $\Prob(B_t^c|B_{t-1},A) \leq \exp\left(-\frac{(m_1-1)n_t}{2m_t}\right) \leq \exp\left(-\frac{(m_1-1)k_{t+2}}{2m_1}\right) = O\left(\exp\left(-n^{0.49}\right)\right)$.
\end{itemize}

Our case analysis thus gives $\Prob(B_t^c|B_{t-1},A) = O\left(\exp\left(-n^{0.49}\right)\right)$. Now we show that all $B_t$'s happen simultaneously with high probability, that is, all filtered sets $V_t$ have an appropriately small number of non-clique vertices. Let $B = \cup_{t=1}^TB_t$. Then $\Prob(B^c|A) \leq \sum_{t=1}^{T} \Prob(B_t^c|B_{t-1},B_{t-2},...,B_1,A)$. But conditioned on the events $B_{t-1},A$, the event $B_t$ is indpendent of $B_1,...,B_{t-2}$. This gives $\Prob(B^c|A) \leq \sum_{t=1}^{T} \Prob(B_t^c|B_{t-1},A) = O\left(T\exp\left(-n^{0.49}\right)\right)$.

\textbf{Step 4:}

If $A$ and $B$ both happen, and $T$ is such that $\frac{m_1n_{T-1}}{m_{T-1}} < \frac{k_0}{2^{T+1}}$, then we have $|V_T \setminus K| = 0$, which means $V_T \subset K$ and $|V_T| \geq \frac{k}{2^{T+3}}$, which is exactly the desired outcome. Note that $\Prob((A,B)^c) \leq \Prob(B^c|A)+\Prob(A^c) = O\left(T\exp\left(-n^{0.49}\right)\right) = O\left(\exp\left(-n^{0.48}\right)\right)$.

So we now only need to show that $\frac{m_1n_{T-1}}{m_{T-1}} < \frac{k_0}{2^{T+1}}$ which is equivalent to $m_{T-1} > 4\sqrt{n}$.

To do this, we need $m_t$ to grow very fast with $t$, and we will show this in steps. First we show that $m_t$ grows with $t$. We then use this growth to show that it grows quite fast. We then use this fast growth to show that it grows very fast.

\begin{enumerate}
\item We prove, by induction on $t$, that for all $t \geq 2$: $m_{t} \geq 4m_{t-1}$ and $m_t \geq 4^{t+1}$.

Recall from the analysis of \textbf{Case 2} that we have assumed $C$ is large enough so that $m_1 \geq 32$. This also means $m_2 = 2^{\frac{m_1}{2}} \geq 4m_1 \geq 128$, which proves the base case for $t=2$. The first inequality holds because the function $2^{t/2}-4t$ is positive for $t=32$ as well as increasing for $t \geq 32$.

Assuming our hypothesis for $2 \leq t \leq \ell-1$, we show it holds for $t =\ell \geq 3$. \footnote{For $\ell =3$, we also use the additional fact $m_1 \geq 32 \geq 16$.} 

Because $m_{\ell-1} \geq 4^{\ell}$, we get $m_{\ell} = 2^{\frac{m_{\ell-1}}{2^{\ell-1}}} \geq 2^{2^{\ell+1}} \geq 2^{2(\ell+1)} = 4^{\ell+1}$.

Since $m_{\ell-2} \geq 4^{\ell-1}$ and $m_{\ell-1} \geq 4m_{\ell-2}$, we get $m_{\ell} = 2^{\frac{m_{\ell-1}}{2^{\ell-1}}} \geq 2^{\frac{4m_{\ell-2}}{2^{\ell-1}}} \geq 2^{2+\frac{2m_{\ell-2}}{2^{\ell-1}}} = 4m_{\ell-1}$.

Note that we have now also shown the fact $m_{t-1} \geq 4^{t-2}m_1$ that we used in \textbf{Step 3}.

\item Now that we have $m_{t} \geq 4m_{t-1}$ for $t \geq 2$, we can show that $m_t$ grows even faster. For $t \geq 2$, 
\[m_{t+1} = 2^{\frac{m_t}{2^{t}}} \geq 2^{\frac{2m_{t-1}}{2^{t-1}}}  = m_t^2.\]

\item Now that we have $m_{t} \geq m_{t-1}^2$ for $t \geq 3$ and $m_t \geq 4$ for $t \geq 1$, we can show that $m_t$ grows much faster. For $t \geq 3$,
\[m_{t+1} = 2^{\frac{m_t}{2^{t}}} \geq 2^{\frac{m_{t-1}}{2^{t-1}}\cdot \frac{m_{t-1}}{2}}  = (\sqrt{m_t})^{m_{t-1}} \geq 2^{m_{t-1}}.\]

\end{enumerate}

Since $\log^*{n}$ is a non-decreasing function, if $t\geq 3$, $\log^*{m_{t+1}} \geq \log^*{(2^{m_{t-1}})} = 1 + \log^*{m_{t-1}}$. Unrolling this gives \[\log^*{m_{t-1}} \geq \frac{t-3}{2} + \log^*{m_2}\] as long as $t$ is odd and $t \geq 3$.

Suppose $\hat{t} := 2\left(\log^*{n}-\log^*{(k/\sqrt{n})}\right)+3$. Because $T \geq \hat{t}$, $m_{T-1} \geq m_{\hat{t}-1}$. Combining this with $m_2 \geq m_1 = k/\sqrt{n}$, the fact that $\log^{*}$ is non-decreasing, and plugging into the inequality above, we get \[\log^*{m_{T-1}} \geq \log^*{n} \implies m_{T-1} \geq n > 4\sqrt{n},\] which completes the proof.

\end{proof}

\vspace{.2cm}
\begin{algorithm}[]
	\SetAlgoLined
	\KwIn{Graph $G = ([n],E) \sim \G(n, \frac{1}{2},k)$, clique size $k$, $t$, vertex $v \in N_t$, access to \textsc{$V_{t-1}$-Membership}}
	\KwOut{Membership in $V_t$ : $\mathbb{1}_{v \in V_t}$}
	Initialize $\text{size}V_{t-1} = 0, \text{deg}V_{t-1}=0$
	
	\For{$u \in N_{t-1}$}{\If{$V_{t-1}\textsc{-membership}(G,k,t-1,u) = 1$ }{
			\hspace{19em}\smash{$\left.\rule{0pt}{2.5\baselineskip}\right\}\ \mbox{Compute $|V_{t-1}|$, `$V_{t-1}$'-degree of $v$}$}
			
			$\text{size}V_{t-1} = \text{size}V_{t-1} + 1$

			$\text{deg}V_{t-1} = \text{deg}V_{t-1} + \mathbb{1}_{(u,v) \in E}$ }
			
		}

		
		\textbf{output} $\mathbb{1}_{\left\{\text{deg}V_{t-1} \geq \frac{\text{size}V_{t-1}}{2} + k_{t+2} -2\sqrt{\text{size}V_{t-1}}\right\}}$

	\caption{\textsc{$V_t$-Membership} ($t \geq 2$)}
	\label{alg:filter}
\end{algorithm}
\vspace{.2cm}

The \textsc{$V_t$-Membership} algorithm simply computes the number of edges from a vertex $v$ to the set $V_{t-1}$ and uses this to determine whether or not $v$ is in $V_t$.
\begin{lemma}[Small space filter implementation]\ \\
\label{lem:smallspacefilterimpl}
Let $G = ([n],E) \sim \G(n, \frac{1}{2},k)$ with a clique size $k$. Let \textsc{$V_1$-Membership} be an algorithm that returns $1$ for every vertex in $N_1$, and let \textsc{$V_t$-Membership} be defined as in Algorithm~\ref{alg:filter} for $t \geq 2$. Given a vertex $v \in N_t$, \textsc{$V_t$-Membership}($G,k,t,v$) returns $1$ if and only if $v \in V_t$. Otherwise it returns $0$. Moreover, it runs in space $O(t \cdot \log n)$.
\end{lemma}
\begin{proof}
We prove this via induction on $t$. For the base case $t=1$, $V_1=N_1$ so the algorithm behaves as advertised. It's space usage is clearly $O(\log n)$ since it outputs a constant.

For the inductive step, we assume the statement of the lemma is true for $t = \ell-1$, and prove it for $t = \ell$. The correctness of \textsc{$V_{\ell}$-Membership} follows immediately from the correctness of \textsc{$V_{\ell-1}$-Membership} and the definition of the vertex set $V_{\ell}$.

Let us now analyse the space usage. To iterate over $N_{t-1}$, the algorithm needs to maintain $u$, and can iterate simply by increasing the name of $u$ by $1$. Additionally, the algorithm also needs to be able to compute $n_{\ell-2},n_{\ell-1}$ to decide the start and end points of the loop. It can do all of this in $O(\log n)$ space because it has access to $n,\ell$ from the input. The algorithm requires a further $O(\log n)$ bits to maintain $0 \leq \text{size}V_{t-1}, \text{deg}V_{t-1} \leq n-1$. Lastly, it needs to run \textsc{$V_{\ell-1}$-Membership}. It can compute $\ell-1$ because it knows $\ell$, and then by our inductive assumption it can run \textsc{$V_{\ell-1}$-Membership} using another $O((\ell-1) \cdot \log n)$ bits of space. Note that this space can be re-used for every call to \textsc{$V_{\ell-1}$-Membership}. Finally, to implement the thresholding, the algorithm also needs access to $k_{\ell+2}$, which it can easily compute in $O(\log n)$ space from the inputs $n,k,\ell$. Square roots can also be computed in logarithmic space upto the desired few bits of precision required to make the comparison. The total space usage is thus $O(\ell \cdot \log n)$ bits, which completes the proof.
\end{proof}

\begin{theorem}\label{thm:main}
Let $G = ([n],E) \sim \G(n, \frac{1}{2},k)$ with a planted clique of size $k \geq C\sqrt{n}$ with the constant $C>0$ chosen as in Lemma~\ref{lem:struct-filt}. Suppose $T : = 2\left(\log^*{n}-\log^*{(k/\sqrt{n})}\right)+3$. Then for large enough $n$, there exists a deterministic algorithm that takes as input the adjacency matrix of the graph and the size of the planted clique, exactly outputs the clique $K$ with probability at least $1-O\left(\left(\frac{1}{n}\right)^{\log k} \right)$ over the randomness in the graph $G$, and runs using $O(T \cdot \log n)$ bits of space.
\begin{enumerate}
\item If $k = C\sqrt{n}$, the space usage is $O(\log^*{n} \cdot \log n)$ bits.
\item If $k = \omega(\sqrt{n}\log^{(\ell)}n)$ for some constant integer $\ell > 0$, the space usage is $O(\log n)$ bits.
\end{enumerate}
\end{theorem}
\begin{proof}
We first note that given $n,k$ as inputs, $T$ can be computed with $O(\log n)$ bits of space. This means we can easily implement an algorithm to check membership in $V_T$. Given a vertex $v \in [n]$, in $O(\log n)$ space we can check if it is in $N_T$. If it is not, we declare it is not in $V_T$. If it is, we run \textsc{$V_T$-Membership}($G,k,T,v$). Due to Lemma~\ref{lem:smallspacefilterimpl} this gives us an $O(T \cdot \log n)$ space oracle that can answer if a vertex is in $V_T$ or not. Moreover, by Lemma~\ref{lem:struct-filt}, except with probability at most $O(\exp(n^{-0.48}))$, $V_T$ is subset of the planted clique and has more than $2 \log n$ vertices. Using this oracle with Algorithm~\ref{alg:completion} (\textsc{Small Space Clique Completion}) and invoking Lemma~\ref{lem:completion} gives us a deterministic algorithm that runs in space $O(T \cdot \log n)$ and outputs the planted clique $K$ with probability at least $1-O\left(\exp(n^{-0.48})+\left(\frac{1}{n}\right)^{\log k} + n\exp\left(\frac{-k}{54}\right)\right) \geq 1-O\left(\left(\frac{1}{n}\right)^{\log k} \right)$ over the randomness in the graph $G$.
\end{proof}

\end{subsection}


\section{Auxiliary Lemmas}
\label{sec:auxlem}
We state the Chernoff bound we use here, for the convenience of the reader.
\begin{lemma}\label{lem:chernoff}
	Let $X = \sum\limits_{i=1}^{n} X_i$ where $X_i$ are independent {\sf{Bern}}($p_i$) random variables. Let $\mu = \sum\limits_{i=1}^{n} p_i$, and $0 < \delta $. Then
	\[\Prob\left( X \geq (1+\delta) \mu \right)  \leq \exp\left(\frac{-\mu \delta^2}{2+\delta} \right)\footnote{If $\delta < 1$, this also means $\Prob\left( X \geq (1+\delta) \mu \right)  \leq \exp\left(\frac{-\mu \delta^2}{3} \right)$, a fact we will use often.}\]
	\[\Prob\left( X \leq (1-\delta) \mu \right)  \leq \exp\left(\frac{-\mu \delta^2}{3} \right)  .\]
\end{lemma}

We state some structural lemmas about the planted clique graph that follow from simple probabilistic arguments.

First we show that with high probability, \textit{any} clique subset of size  greater than $2\log n$ has at most $3\log k$ non-clique vertices connected to every vertex of the subset. The ideas of such an analysis are contained in the proof of~\cite[Lemma
2.9]{dekel2014finding}.\\

\begin{lemma}\label{lem:false=pos}
	Let $G  \sim \G(n, \frac{1}{2},k)$ for $k \geq 2 \log n$ and $S$ be any arbitrary subset of the planted clique $K$ with $|S| \geq 2 \log n$. Let $T$ be the set of all non-clique vertices that are connected to every vertex in $S$. Then, except with probability at most $\left(
	\frac{1}{n}\right) ^{\log k}$, $|T| \leq
	3\log k$.
\end{lemma}
\begin{proof}
	Fix $S' \subset S \subset K$ such that $\lvert S' \rvert = 2\log{n}$. Let $T'$ denote the set of all non-clique vertices that are connected to every vertex in $S'$. Clearly, $|T| \leq |T'|$. So we will show that $|T'| \leq
	3\log k$ except with probability at most $\left(
	\frac{1}{n}\right) ^{\log k}$.
	
	Let $W$ be any subset of $K$ with $|W| = 2 \log n$. The probability there exists a subset of non-clique vertices of size
	$\ell$ connected to every element in $W$ is at most $\binom{n}{\ell}2^{-\ell
		(2 \log{n})}$.  A union bound then implies that the probability there exists a
	subset of non-clique vertices of size at least $\ell_0 =
	1+3\log k$ connected to every
	element in $W$ is at most $\sum_{\ell = \ell_0}^{n - k} \binom{n}{\ell}2^{-\ell
		(2\log{n})} \leq 2^{-3 \log k \log{n}}$.  Further union bounding over all
	subsets of $K$ of size $2\log{n}$ implies $|T'| \leq
	1+3\log k$ except with probability at most
	\[
	\binom{k}{2\log{n}}2^{-3 \log k \log{n}} \leq
	2^{2\log k \log n}2^{-3 \log k\log n} =  2^{- \log k \log n} = \left(
	\frac{1}{n}\right) ^{\log k}.
	\]
\end{proof}

We also control the number of clique vertices
any non-clique vertex is connected to.
\begin{lemma}\label{lem:2kby3}
	Let $G \sim \G(n, \frac{1}{2},k)$, and let $d$ be the maximum number of clique vertices connected to a non-clique vertex. Then $\Prob(d \geq \frac{2k}{3}) \leq n \exp\left(\frac{-k}{54}\right)$.
\end{lemma}
\begin{proof}
	A Chernoff bound
	(Lemma~\ref{lem:chernoff}) shows that any given non-clique vertex has is connected to more than $\frac{2k}{3}$ clique vertices with probability at most $\exp\left(\frac{-k}{54}\right)$ and a union bound over the at most $n$ non-clique vertices then finished the proof.
\end{proof}

For the convenience of the reader, we present a proof of the well known fact that Erd\H{o}s-R\'enyi graphs do not have large cliques. See, for example, \cite{bollobas1976cliques}.
\begin{lemma}\label{lem:gnpcliquesize}
Let $G \sim \G(n, \frac{1}{2})$ and $\epsilon >0$ be a positive constant. Except with probability at most $O\left(2^{-\epsilon \log^2 n}\right)$, $G$ contains no cliques of size $(2+\epsilon) \log n$ or larger.
\end{lemma}
\begin{proof}
If $G$ has a clique of size larger than $(2+\epsilon) \log n$, it also has a clique of size $(2+\epsilon) \log n$. By a simple union bound over all vertex subsets of size $(2+\epsilon) \log n$, the probability that $G$ has a clique of this size is at most ${n \choose (2+\epsilon) \log n}2^{-{(2+\epsilon) \log n \choose 2}} = O\left(2^{-\epsilon \log^2 n}\right)$.
\end{proof}

We show the existence of a $O(\log^2 n)$ space recovery algorithm above the information theoretic threshold.
\begin{lemma}[\cite{alon2007testing} reduction + $O(\log^2 n)$ space detection]\ \\
	\label{lem:log2n-recovery}
	Let $\omega(\log n)=k=o(n)$ and $G \sim \G(n, \frac{1}{2}, k) = ([n],E)$. Then there is a deterministic $O(\log^2 n)$ space algorithm that outputs the planted clique except with probability at most $O(n\exp\left(-k/54\right)+n2^{- \Theta(\log^2 n)})$.
\end{lemma}
\begin{proof}
For a vertex $v \in [n]$, denote by $G_v$ the graph induced on the vertex subset formed by removing $v$ and all its neighbours from $[n]$. Assume that every non-clique vertex in $G$ is connected to at most $\frac{2k}{3}$ clique vertices. By Lemma~\ref{lem:2kby3}, this happens except with probability at most $n\exp\left(-k/54\right)$. Further assume that every vertex in $G$ has degree at most $2n/3$. By a union and Chernoff bound, this happens except with probability at most $n\exp(-cn)$ for some constant $c>0$. By a union bound, we can assume that both the structural assumptions we have made hold simultaneously except with probability at most $O(n\exp\left(-k/54\right))$.
	
This means that if $v$ is a clique vertex, $G_v$ is an Erd\H{o}s-R\'enyi graph with no planted clique and at least $n/3$ vertices. By a further union bound and using Lemma~\ref{lem:gnpcliquesize}, we assume the largest clique in $G_v$ for all clique vertices $v$ is less than $3 \log n$. Overall, all our structural assumptions hold except with probability at most $O(n\exp\left(-k/54\right)+n2^{- \Theta(\log^2 n)})$.
	
If $v$ is not a clique vertex, $G_v$ is a planted clique graph with a planted clique of size at least $k/3$. Hence it has a clique of size $3\log n$. We can use this property to distinguish between clique and non-clique vertices.

Our algorithm can use a $O(\log n)$ bit counter to loop over all vertices in $[n]$. For a given vertex $v$, our algorithm says it is not in the planted clique if and only if it finds a clique of size $3 \log n$ in $G_v$. To check this, the algorithm can store $3\log n$ names of vertices (taking $O(\log^2 n)$ bits of space) and loop over all possibilities. If it finds a clique formed by vertices that are all unconnected to $v$, it declares $v$ to be not in the planted clique. To check the existence of a clique for a given set of $3 \log n$ vertices, it only needs a further $O(\log\log n)$ bits of space to loop over all possible edges between this set of vertices. The overall space usage is thus $O(\log^2 n)$ bits.

\end{proof}

\section*{Acknowledgments}
We would like to thank Dean Doron, G\'abor Lugosi, and Kevin Tian for helpful discussions and pointers to relevant literature.



\addcontentsline{toc}{section}{References}
\bibliographystyle{alpha}
\bibliography{ref}




\end{document}